\documentclass[11pt]{article}  
\usepackage{graphicx, xcolor}
\usepackage{amsmath, amsthm, amsfonts, enumerate, amssymb, setspace, booktabs, threeparttable}
\usepackage{acronym}
\usepackage{algorithm}
\usepackage[noend]{algpseudocode}

\usepackage{lineno}
\usepackage[colorlinks=true,allcolors=blue]{hyperref}

\usepackage[fancythm,fancybb]{jphmacros} 

\newcommand{\ba}{\begin{align}}
\newcommand{\ea}{\end{align}}

\newcommand{\lam}{\lambda}

\usepackage{algorithm}
\usepackage{algpseudocode}
\newcommand\ignore[1]{{}}

\theoremstyle{definition}

\usepackage{mathtools}

\newcommand\Algphase[1]{%
\vspace*{-.7\baselineskip}\Statex\hspace*{\dimexpr-\algorithmicindent-2pt\relax}\rule{\columnwidth}{0.4pt}%
\Statex\hspace*{-\algorithmicindent}\textbf{#1}%
\vspace*{-.7\baselineskip}\Statex\hspace*{\dimexpr-\algorithmicindent-2pt\relax}\rule{\columnwidth}{0.4pt}%
}

\usepackage{graphicx, xcolor}
\usepackage{multirow}
\usepackage{color, colortbl}
\definecolor{RoyalBlue}{rgb}{0.9,1,1}

\acrodef{esc}[ESC]{extremum seeking control}
\acrodef{lmis}[LMIs]{linear matrix inequalities}

\title{\sc Safe Learning-based Observers for Unknown Nonlinear Systems using Bayesian Optimization}
\author{Ankush Chakrabarty$^{\dag}$ and Mouhacine Benosman%
\thanks{$^1$All authors are affiliated with Mitsubishi Electric Research Laboratories, Cambridge, MA, USA.}%
\thanks{$^\dag$Corresponding author: A.~Chakrabarty. Phone: +1~(617)~758-6175. Email: \texttt{achakrabarty@ieee.org}}}
	
\newcommand{\bestrewardk}{\hat\costfn_j^\star}

\newcommand{\costfn}{\mathcal J}

\newcommand{\lc}{\mathfrak{L}_\phi}
\newcommand{\lchat}{\hat{\mathfrak L}_\psi}
\newcommand{\lpsi}{\mathfrak{L}_\psi}

\date{\today}
\usepackage[margin=1.25in]{geometry}

\begin{document}
	
	\maketitle
    \begin{abstract}
	{\sf Data generated from dynamical systems with unknown dynamics enable the learning of state observers that are: robust to modeling error, computationally tractable to design, and capable of operating with guaranteed performance. In this paper, a modular design methodology is formulated, that consists of three design phases: (i) an initial robust observer design that enables one to learn the dynamics without allowing the state estimation error to diverge (hence, safe); (ii) a learning phase wherein the unmodeled components are estimated using Bayesian optimization and Gaussian processes; and, (iii) a re-design phase that leverages the learned dynamics to improve convergence rate of the state estimation error. The potential of our proposed learning-based observer is demonstrated on a benchmark nonlinear system. Additionally, certificates of guaranteed estimation performance are provided.}
	\end{abstract}

	\section{Introduction}
	With modern dynamical systems growing in complexity, data-driven methods of state estimation have gained relevance and provide a suitable solution for problems with incomplete model knowledge, such as systems with unknown or unmodeled dynamics. Leveraging data generated from these systems and employing function approximators (such as neural networks) allows one to design controllers and estimators despite incomplete model descriptions, by identifying the unknown dynamics. While the most common applications of learning via function approximators are in identifying dynamical systems from measurement data~\cite{narendra1990identification,kosmatopoulos1995high}, approximating control laws using offline data~\cite{chakrabarty2017support,drgovna2018approximate}, or adapting with data to generate optimal control policies online~\cite{kiumarsi2018optimal,huang2014neural}, the utility of data and approximators for state estimation in nonlinear systems with unmodeled dynamics remains relatively unexplored.

Some early investigations into data-driven observers, for example, in~\cite{theocharis1994neural} assume model availability. However, the current wave of data-driven control has demonstrated the effectiveness of approximators in controlling systems with partial models or, in some cases, a `model-free' manner. Data-driven observers in the model-free setting were explored almost two decades ago in~\cite{Kim1997}, where the authors proposed an adaptation rule for learning the weights of a linear-in-parameter neural network (LPNN) that results in uniformly ultimately bounded estimation error dynamics. Although this work has been adopted in multiple applications such as robot control~\cite{Abdollahi2006,Liu2018}, rotors~\cite{Pratap2010}, and more recently, wind turbines~\cite{Rahimilarki2019}, the inherent assumptions and theory have hardly evolved. In most of these methodologies, the activation functions are considered to be radial basis functions, there is no measurement noise, and the theoretical guarantees of learning performance remain the same; an exception is~\cite{Rahimilarki2019} where the authors investigate input-to-state stability (ISS) observers for the known component of the model, but the learner performance is not ISS, and the learner's weights require manual tuning.

In~\cite{benosman2019robust}, the authors studied the problem of learning-based observer design for a class of partial differential equations (PDEs). The idea is based on designing a robust controller with respect to the structured model uncertainties, complemented with a learning layer, to auto-tune the observer gains, using a data-driven extremum seeking algorithm. However, the model uncertainties were compensated for by a robust estimation approach, and no online estimation of the model was proposed in this work.
In~\cite{koga2019learning}, the authors proposed a learning-based observer design for the particular case of the Boussinesq PDE equations with parametric uncertainties, for a thermo-fluid application. The main idea is to construct a nominal observer which ensures ISS between the state estimation error and the parameter estimation error. The parameters are then estimated online using an extremum seeking method. The result was limited to the specific case of Boussinesq equations with structured parametric uncertainties.

In \cite{chakrabarty2019cdc}, the problem of adaptive observer design for a class of nonlinear systems was studied. The authors proposed to use shallow neural nets, with Lipschitz activation functions, to estimate the unknown parts of the model. The use of Lipschitz constants of these activation functions simplifies the data-driven observer design procedure via new LMI conditions, ensuring pre-computable robust stability guarantees on the resulting state estimation error. The learning results, however, are model-based in the sense that the adaptation laws designed to learn the neural net coefficients are based on the model structure, similar to an indirect model-based adaptation approach, which limits the validity of the learning laws to a specific class of models, see for example~\cite{ben_survey18} for a discussion about model-based and model-free learning and adaptation. 

In~\cite{na2017adaptive}, the authors propose to formulate the problem of observer design for a class of partially known linear and nonlinear systems, as an optimal problem with quadratic cost over an infinite time support. Then the authors propose to use adaptive dynamic programming to solve this equivalent optimal control problem. However, this leads to a bounded integral cost of the estimation error, rather than, the usual sought for convergence result to a zero estimation error. This paper also requires a specific structure on the output matrix that is not always available.

In~\cite{capone2019interval}, the authors propose, in the context of interval observers, a Gaussian process-based observer for a class of partially known nonlinear autonomous systems. The Gaussian process is used to approximate the unknown part of the model, which assumed to be a Lipschitz continuous function. 

\subsubsection*{Contributions}
In this paper, we extend the ideas from \cite{chakrabarty2019cdc,chakrabarty2019estimating,koga2019learning} to the general case of partially known nonlinear models, where we propose to approximate the unknown part using neural networks, under a robustness design constraint in the form of local input-to-state stability between the state estimation error and the model parametric estimation error. This approach allows us to obtain a \textit{safe} learning of the model, in the sense of maintaining boundedness of the state estimation error at all time, even during learning. After convergence of the learning algorithm, we propose to add a redesign step, which takes into account the learned model to improve the observer performance. 

\subsubsection*{Organization}
The rest of the paper is organized as follows. Section~\ref{sec:ps} describes the overall problem we are solving in this paper, along with assumptions. Section~\ref{sec:obs_nominal} introduces a method for generating initial robust observer gains, and a method for learning the unmodeled dynamics via Bayesian optimization is discussed in Section~\ref{sec:bayesopt}. An optional re-design of the observer gain is presented in Section~\ref{sec:redesign} for improving the convergence rate by exploiting the learned nonlinearity, and the potential of our proposed approach is demonstrated using a numerical example in Section~\ref{sec:ex}. Conclusions are drawn in Section~\ref{sec:conc}.
\subsubsection*{Notation}
We denote by $\mathbb{R}$ the set of real numbers, $\mathbb R_+$ as the set of positive reals, and $\mathbb{N}$ as the set of natural numbers. In a metric space $(\mathbb X, \rho)$, we denote an open ball centered at $x_0\in\mathbb X$ with radius $\varepsilon$ as $\mathbb B_\varepsilon(x_0) := \{x\in\mathbb X \,|\, \rho(x, x_0) < \varepsilon\}$. The symbols $\mathsf E(\cdot)$ and $\mathsf V(\cdot)$ denote the expectation and variance operators of a random variable, respectively. The notation $\ln$ represents the natural logarithm. A function $f(x)\in \mathcal C^n$ if it is $n$-times differentiable and all its derivatives are continuous, and $f$ is $L$-Lipschitz continuous (or $L$-Lipschitz, for brevity) if $\|f(x)-f(y)\|\le L\|x-y\|$ for some $L>0$ and any $(x,y)$ pair in its domain. A function $f\in\mathcal L_\infty$ if $\|f(\cdot)\|<\infty$ on its domain. A continuous function $f:[0,\infty)\to [0,\infty)$ belongs to class $\mathcal K$ if it is strictly increasing and $f(0)=0$ and of class $\mathcal K_\infty$ if it is of class $\mathcal K$ and unbounded (with range extended to $[0,\infty]$). A continuous function $f$ belongs to class $\mathcal{KL}$ if, for each fixed $t$, the mapping $f(t, x)$ belongs to class $\mathcal K$ with respect to $t$ and, for each fixed $x$, the mapping $f(t,x)$ is decreasing with respect to $t$ and $f(t,x)\to 0$ as $t\to\infty$. For a square matrix $M$, its largest and smallest eigenvalue are denoted by $\lam_{\max}(M)$ and $\lam_{\min}(M)$, respectively. The transpose of a real-valued matrix $M$ is denoted $M^\top$, and its Frobenius norm is denoted $\|M\|$. The quadratic form is written briefly as $\|x\|_Q^2:= x^\top Q x$.

	\section{Problem Statement}\label{sec:ps}
	\subsection{Preliminaries}
	We consider nonlinear systems modeled by
	\begin{subequations}
		\label{eq:system_description}
		\begin{align}
			x_{t+1} &= Ax_t + Bu_t + \phi(q_t),\\
			y_t &= Cx_t,\\
			q_t &= C_q x_t,
		\end{align}
		where $t\in\mathbb R$ is the time index, $x\in\mathbb R^{n_x}$ denotes the system state, $u\in\mathbb R^{n_u}$ the known control input, $y\in\mathbb R^{n_y}$ the measured output, and $q_t\in\mathbb Q\subset \mathbb R^{n_q}$ the argument of the nonlinearity $\phi$. The pair $(A,B)$ is stabilizable, the pair $(A,C)$ is observable, and the set of arguments $\mathbb Q$ is compact. Our knowledge of the system is summarized in the following assumption.
		\begin{assumption}\label{asmp:knowledge_dynamics}
			The system matrices $A$, $B$, and $C$ are known. The nonlinearity's argument matrix $C_q$ is known. The nonlinearity $\phi$ is completely unknown.
		\end{assumption}

		We also add a boundedness and Lipschitz assumption on the nonlinearity.
\begin{assumption}\label{assumption_phi_bound}
The unknown nonlinearity is bounded: there exists a scalar $\overline{\phi}>0$ such that $\phi(q)\leq \overline{\phi}$ for all $q\in \mathbb Q$. Furthermore, $\phi$ is Lipschitz continuous on $\mathbb Q$.
\end{assumption}
\begin{assumption}\label{asmp:x_bdd}
The discrete Cauchy problem (\ref{eq:system_description}) admits a bounded\footnote{as defined in \cite{RD03}, Definition 1.1} solution for any initial condition $x_0\in\mathbb R^{n_x}$.
\end{assumption}
Since every Lipschitz continuous nonlinearity is \textit{a fortiori} continuous, and the domain of $\phi$ (namely, $\mathbb Q$) is compact by assumption, one can use a wide range of function approximators to express $\phi$ as a basis expansion. These approximators need only induce functions that are dense in the space of continuous functions such as polynomials (recall the Stone-Weierstrass Theorem) or `universal approximators' like shallow/deep neural networks with appropriate activation functions~\cite{hanin2019universal}. Therefore, we can rewrite the nonlinearity $\phi$ in the basis expansion form
\begin{equation}\label{q_func_psi}
\phi(q) = p_\star^\top \psi(q) + \epsilon_{\phi}(q),
\end{equation}
\end{subequations}
where $p_\star$ denotes a $N_b \times n_x$ matrix of correct coefficients of the basis expansion, $\psi$ denotes a $N_b$-dimensional vector of $\mathcal C^1$ basis/activation functions, and $\epsilon_{\phi}$ denotes a bounded approximation error; that is, $\|\epsilon_{\phi}(q)\|\leq \bar\epsilon_\phi$ for any $q\in\mathbb Q$. This basis function form is also referred to in the literature as a neural approximator~\cite{na2017adaptive}. We note that the $\mathcal C^1$ nature of $\psi$ implies that it is Lipschitz continuous on $\mathbb Q$; we denote it's Lipschitz constant as $\lpsi$. Note that by equation (\ref{q_func_psi}), together with Assumption \ref{assumption_phi_bound}, and the boundedness of $p_\star$ leads to the boundedness of $\psi$ in $\mathbb Q$.
\begin{assumption}\label{asmp:p_bdd}
There exists a known positive scalar $\bar p_\star$ such that $\|p^\star\|\le \bar p_\star$.
    \end{assumption}
    Assumption~\ref{asmp:p_bdd} is admittedly strong, but a good estimate of $\bar p_\star$ is not needed, as we shall see in the remainder of the paper. In fact, this is an assumption to ensure good robustness properties of the initial observer. As data is obtained, the learning will overcome the conservativeness of the assumption and $\bar p_\star$ will be computed (not guessed) for the redesign.
    
    This motivates a state observer of the form
    \begin{subequations}\label{eq:obs}
    \begin{align}
    \hat x_{t+1} &= A\hat x_t + B u_t + p_t^\top \psi(\hat q_t) + L_0(C\hat x_t - y_t),\\
    \hat q_t &= C_q \hat x_t,
    \end{align}
    \end{subequations}
    with gain matrix $L_0\in\mathbb R^{n_x \times n_y}$. The error dynamics of the observer~\eqref{eq:obs} with $e_t:= \hat x_t - x_t$ and $e^p_t=p_t-p_\star$ is given by
    \begin{align}\label{eq:ester1}
    e_{t+1} = A_L e_t + \psi(\hat q_t)^\top e_t^p + p_\star^\top\Delta \psi_t - \epsilon_\phi(q_t),
    \end{align}
where $A_L = A+L_0 C$ and $\Delta\psi_t=\psi(\hat{q}_t)-\psi(q_t)$.
    Our initial step is to design a learning algorithm and the observer gain $L_0$ 
    such that
    \[
    {e_t \to \mathbb B_{\varepsilon}(0) \;\text{and}\; p_t \to p_\star \;\text{as}\; t\to\infty,}
    \]
	and then switch to a more precise observer, once we have learned the model uncertainties. Indeed, our data-driven observer will operate in three design modes: (i) in the data collection mode, we will design $L_0$ such that the error dynamics are locally input-state stable (L-ISS) between the parameter estimation error $p_t^\top  - p_\star $ and the state estimation error $e_t$, setting the stage for safe learning; (ii) in the safe learning mode, we will leverage Bayesian optimization to learn the coefficients $p_t$; and, (iii) as the learning  terminates, we will redesign $L_0$ taking the learned nonlinearity into account to improve the estimation performance of the observer, based on the learned model.  
	
	\section{Nominal ISS observer design}
	\label{sec:obs_nominal}
	The main idea of the nominal observer is to design the gain $L_0$ such that one can guarantee a local ISS property between the parameter estimation error and the state estimation error. In this section, we present conditions that, if satisfied, enable the design of a suitable $L_0$. Consequently, we propose a simple convex programming formulation for computing an $L_0$ that satisfies these conditions.
	\subsection{Conditions for initial design}
    We begin with the following definitions from~\cite{Li2016}.
    \begin{definition}[L-ISS]\label{def:lISS}
    A discrete-time nonlinear system
    \begin{equation}\label{eq:exemplar_sys}
    \xi_{t+1} = \varphi(\xi_t, \nu_t)
    \end{equation}
    is locally input-to-state stable (L-ISS) with respect to the exogenous input $\nu_t$, if there exist scalars $\varrho_\xi > 0$, $\varrho_\nu > 0$, $\gamma\in\mathcal K_\infty$ and a function $\beta(\cdot,\cdot)\in \mathcal{KL}$ such that for all $\|x_0\| \le \varrho_\xi$ and $\|\nu\|_\infty\le \varrho_{\nu}$, the condition
    \[
    \|\xi(t, \xi_0,\nu)\| \le \beta(t,\|\xi_0\|) + \gamma(\|\nu\|_\infty)
    \]
    is satisfied, for all $t\in\mathbb N$.
    \end{definition}
    
    \begin{definition}[L-ISS Lyapunov function]\label{def:localISS_lyap}
    Let the set $\mathcal G \subset \mathbb R^n$ contain the origin in its interior. A Lipschitz continuous function $V : \mathcal G \to [0,\infty)$ is called a local ISS Lyapunov function for a system~\eqref{eq:exemplar_sys}
    on $\mathcal G$ if there exist $\mathcal K_\infty$ functions $\gamma_1(\cdot)$, $\gamma_2(\cdot)$, $\gamma(\cdot)$ and $\beta(\cdot)$ such that
    \begin{subequations}
    \begin{align}\label{eq:lyapISS_1}
    \gamma_1(\|\xi\|) \le V(\xi) &\le \gamma_2(\|\xi\|)\\\label{eq:lyapISS_2}
    V(\varphi(\xi, \nu)) - V(\xi) &\le      -\gamma(\|\xi\|)+ \beta(\|\nu\|)
    \end{align}
    \end{subequations}
    is satisfied, for all $\xi\in\mathcal G$ and $\nu\in \Xi$.
    \end{definition}
    
    \begin{lemma}\label{lem:localISS}
    If the system~\eqref{eq:exemplar_sys} admits a locally ISS Lyapunov function, then it is locally ISS w.r.t. $\nu$.
    \end{lemma}
    \begin{proof}
    See~\cite[Theorem 4]{jiang2001input} .
    \end{proof}
	
	The following theorem provides conditions for designing $L_0$ such that $(A+L_0C)$ is Schur (stable), and the state estimation error $e$ is L-ISS with respect to the parameter estimation error $e^p$ and the neural approximation error $\epsilon_\phi$.
	\begin{theorem} \label{thm1} 
	Recall $A_L = A+L_0 C$. Let $\lpsi:=\max_{\mathbb Q} \|\nabla_q \hat\psi\|$ and Assumptions 1--\ref{asmp:p_bdd} hold.  
    Suppose there exist matrices $P=P^\top \succ 0$, $Q=Q^\top \succ 0$, and an initial observer gain $L_0$ such that
	\begin{subequations}\label{eq:thm1_all}
	\begin{align} 	\label{eq:lyap_eqn}
	A_L^\top P A_L - P + Q &= 0,\\
	\label{eq:gain_cond} 
4\lam_{\max}(P)\bar p_\star^2 \lpsi^2 \|C_q\|^2 + 8\bar p_\star\lpsi\|PA_L\|\|C_q\| &\leq \lam_{\min}(Q)
	\end{align}
\end{subequations}
then the error system~\eqref{eq:ester1} is L-ISS w.r.t. $\left[{e_t^p}^\top, {\epsilon_\phi}^\top\right]^\top$.
\end{theorem} 
\begin{proof}
We consider a candidate of L-ISS Lyapunov function $V$ defined by
\begin{align}
\label{eq:Lyapdef-G}
V_t = e_t^\top P e_t,
\end{align} 
For quadratic $V$, the condition~\eqref{eq:lyapISS_1} is satisfied with $\gamma_1=\lambda_{\min}(P)\|e\|^2$ and $\gamma_2=\lambda_{\max}(P)\|e\|^2$. Taking the time difference $\Delta V_t:= V_{t+1} - V_t$ of~\eqref{eq:Lyapdef-G} along the solution of \eqref{eq:ester1}, we obtain 
\begin{equation} 
\Delta V_t = \sum_{\ell=1}^6 v_\ell,
\end{equation}
where
\begin{align*}
v_1 &= e_t^\top(A_L^\top P A_L- P)e_t = - e_t^\top Q e_t,\\
v_2 &= (p_\star^\top\Delta\psi_t+\psi(\hat{q})^\top e_t^p)^\top P (p_\star^\top\Delta\psi_t+\psi(\hat{q})^\top e_t^p), \\
v_3 &= \epsilon_{\phi}^\top P \epsilon_{\phi},\\
v_4 &= 2(A_Le_t)^\top P (p_\star^\top\Delta\psi_t+\psi(\hat{q})^\top e_t^p),\\
v_5 &= -2 \epsilon_\phi^\top P A_L e_t,\\
v_6 &= -2\epsilon_\phi^\top P (p_\star^\top\Delta\psi_t+\psi(\hat{q})^\top e_t^p).
\end{align*}
where the rightmost equation of $v_1$ is obtained by replacing with $Q$ using~\eqref{eq:lyap_eqn}.

Before further analysis, note that we can bound the norms of $p_\star^\top\Delta\psi_t$ and $\psi(\hat{q})^\top e_t^p$ as follows:
\begin{align*}
\|p_\star^\top \Delta \psi_t\| &\le \bar p_\star\lpsi\|q_t - \hat q_t\|\le \bar p_\star\lpsi\|C_q\|\|e_t\|,\\
\|\psi(\hat{q})^\top e_t^p\|&\le \|\psi(\hat{q})\| \|e_t^p\| = \|\bar\psi\|\|e_t^p\|,
\end{align*}
where the first bound is a consequence of the Lipschitz property of $\psi$ and the second due to $\mathbb Q$ being compact and $\psi$ being bounded on $\mathbb Q$. Using these bounds, we can proceed with bounded the $v_\ell$ terms. Concretely,
\begin{align*} 
v_1 &\le -\lam_{\min}(Q)\|e_t\|^2,\\
v_2 &\le \lam_{\max}(P)\left(\bar p_\star^2 \lpsi^2\|C_q\|^2\|e_t\|^2 + \bar\psi^2\|e_t^p\|^2\right),\\
v_3 &\le \lam_{\max}(P)\|\epsilon_\phi\|^2,\\
v_4 &\le 2\bar p_\star\lpsi\|PA_L\|\|C_q\|\|e_t\|^2 + 2\bar\psi\|PA_L\|\|e_t^p\|\|e_t\|,\\
v_5 &\le 2\bar\epsilon_\phi\|PA_L\|\|e_t\|,\\
v_6 &\le 2\bar p_\star\bar\epsilon_\phi \lpsi\|P\|\|C_q\|\|e_t\| + 2\bar\psi\|e_t^p\|.
\end{align*}
Collecting relevant terms, we can write
\begin{equation} 
\Delta V_t \leq -a_0\|e_t\|^2 + a_1\|e_t\|^2 + a_2\|e_t\|+a_3\|e_t^p\|^2 + a_4\|\epsilon_\phi\|^2,
\notag\end{equation} 
where 
\begin{align*}
a_0 &= \tfrac{1}{2}\lambda_{\min}(Q) > 0,\\
a_1 &= -\tfrac{1}{4} \lam_{\min}(Q) + \lam_{\max}(P)\bar p_\star^2 \lpsi^2 \|C_q\|^2 + 2\bar p_\star\lpsi\|PA_L\|\|C_q\|,\\
a_2 &= -\tfrac{1}{4}\lam_{\min}(Q)\|e_t\| + (2\|PA_L\|+2\bar p_\star\lpsi\|P\|\|C_q\|)\bar\epsilon_\phi \\ &\hspace{5em} + 2\bar\psi\|PA_L\|\|e_t^p\|,\\
a_3 &= \lambda_{\max}(P)\bar\psi^2 > 0,\\
a_4 &= \lambda_{\max}(P) >0.
\end{align*}
Since the observer gain $L_0$ satisfies the condition~\eqref{eq:gain_cond}, we get $a_1 \le 0$. Therefore, if $a_2 \le 0$, that is, if 
\begin{align}\notag 
\|e_t^p\|\le \frac{\lam_{\min}(Q)}{8\bar\psi\|PA_L\|}\|e_t\| - \frac{\|PA_L\|+\bar p_\star\lpsi\|P\|\|C_q\|}{\bar\psi\|PA_L\|}\|\epsilon_\phi\|,
\end{align}
then
\begin{equation}\label{perform_quant_0}
\Delta V_t \leq  - a_0 \| e_t\|^2+ a_3\|e_t^p\|^2 + a_4\|\epsilon_\phi\|^2, 
\end{equation}
which satisfies the condition~\eqref{eq:lyapISS_2} with $\nu$ replaced by $\left[{e_t^p}^\top, {\epsilon_\phi}^\top\right]^\top$.
\end{proof} 
\subsection{Computing the initial observer gain}

We begin with the following theorem that provides a design procedure for $L_0$ for a fixed $\lpsi$. Such an $\lpsi$ could be known, or, if unknown, one could design an observer for a maximal $\lpsi$ (obtained via a line search) such that the conditions~\eqref{eq:thm1_all} are satisfied. The following convex relaxation is proposed.

\begin{theorem}
For fixed $\bar p_\star$ and $\lc$, if there exist matrices $P=P^\top$, $Q=Q^\top$, $K\in\mathbb R^{nx\times ny}$, and scalars $\beta_\kappa>0$, $\kappa_1>0$, $\kappa_2>0$, and $\kappa_3>0$ such that
\begin{subequations}
\label{eq:thm2_all}
\begin{align}
&\min_{P, K, Q, \kappa_0, \kappa_1, \kappa_2, \kappa_3} \kappa_0 + \beta_\kappa(\kappa_1 +\kappa_2 + \kappa_3)\label{eq:thm2_a}\\
\text{subject to:} &\notag\\
&\begin{bmatrix}
\label{eq:thm2_b}
-P + Q + \kappa_0 I & \star \\ PA + KC & P
\end{bmatrix} \preceq 0\\
\label{eq:thm2_f}& Q \succeq \kappa_1 I\\
\label{eq:thm2_c}& 0 \prec P \preceq \kappa_2 I\\
\label{eq:thm2_d}& \|PA+KC\| \le \kappa_3\\
\label{eq:thm2_e}& 
4\bar p_\star^2 \lpsi^2 \|C_q\|^2\kappa_2 + 8\bar p_\star\lpsi\|C_q\|\kappa_3 -\kappa_1 \le 0
\end{align}
\end{subequations}
has an optimal solution with $\kappa_0=0$, then an observer of the form~\eqref{eq:obs} with gain $L_0 = P^{-1}K$ yields error dynamics~\eqref{eq:ester1} that are L-ISS.
\end{theorem}

\begin{proof}
We begin by noting $K:=PL_0$ and substituting $K$ into~\eqref{eq:thm1_all}. This yields
\begin{equation}\label{eq:thm2_pf2}
(PA + KC)^\top P^{-1} (PA + KC) - P + Q = 0,
\end{equation}
\begin{equation}\label{eq:thm2_pf0}
4\lam_{\max}(P)\bar p_\star^2 \lpsi^2 \|C_q\|^2 + 8\bar p_\star\lpsi\|PA+KC\|\|C_q\| \leq \lam_{\min}(Q).
\end{equation}
We relax the equality~\eqref{eq:thm2_pf2} with the inequality
\begin{equation}\label{eq:thm2_pf1}
(PA + KC)^\top P^{-1} (PA + KC) - P + Q \preceq -\kappa_0 I.
\end{equation}
Taking Schur complements of~\eqref{eq:thm2_pf1} yields~\eqref{eq:thm2_b}. Conditions~\eqref{eq:thm2_f}--\eqref{eq:thm2_d} provide bounds on $\lam_{\min}(Q)$, $\lam_{\max}(P)$, and $\|PA+KC\|$: with these bounds $\kappa_1$, $\kappa_2$, and $\kappa_3$, one can explicitly write the constraint~\eqref{eq:thm2_pf0} as~\eqref{eq:thm2_e}. If the optimal value of $\kappa_0=0$, then the equality~\eqref{eq:thm2_pf2} is exact, and the conditions in~\eqref{eq:thm1_all} are satisfied.
\end{proof}

\begin{remark}
Note that, for a fixed $\lc$ and $\bar p_\star$, the problem~\eqref{eq:thm2_all} is convex, and therefore can be solved efficiently using standard convex solvers such as CVX/YALMIP.
\end{remark}

\begin{remark}
	The scalar $\beta_{\kappa}$ is a regularization parameter that should be kept small enough to promote the computation of unique solutions while ensuring the focus of the objective function is to force $\kappa\approx 0$.
\end{remark}

\begin{remark}
One can use the conditions in~\eqref{eq:thm2_all} to perform a grid search for the largest Lipschitz constant and the largest coefficient bound for which a feasible solution to~\eqref{eq:thm1_all} exists. Concretely, by solving the problem
\begin{equation}\label{eq:linesearch}
(\lchat,\hat{\bar{p}}_\star) :=\argmax (\lpsi + \bar{p}_\star) \text{~subject to:~\eqref{eq:thm2_all}}
\end{equation}
and ensuring $\kappa_0 = 0$ for the optimal solution of~\eqref{eq:linesearch} generates $\lchat$ and $\hat{\bar{p}}_\star$ and $L_0$ such that the error dynamics~\eqref{eq:ester1} are L-ISS, as long as these are overestimates; i.e. $\lchat\ge \lpsi$ and $\hat{\bar{p}}_\star\ge \bar p_\star$. Of course, if either one is known \emph{a priori}, then one can perform a line search for the other.
\end{remark} 

\section{Learning via Bayesian optimization}\label{sec:bayesopt}
The previous section described a method for starting with a `safe' initial state estimator - one whose error trajectory does not diverge despite learning, as long as the coefficients $p_t$ of the parametric model are not unbounded. In this section, we provide a data-driven method for updating the model from on-line data using Bayesian optimization with Gaussian process surrogate modeling and the expected improvement acquisition function.

\subsection{Data collection and reward}
Let $T_\ell$ denote the horizon over which measurements and state estimates are collected, and $p_0$ denote the initial guess for the model coefficients. We begin with no data, and an initial observer gain $L_0$. We run the observer~\eqref{eq:obs} over the learning horizon $\{0,1,\ldots, T_\ell -1\}$ and compute the reward
\begin{equation}\label{eq:costfn_comp}
\mathcal J(p_j) := -\left(\sum_{t=0}^{T_\ell - 1} \left\|C\hat x_t - y_t\right\|_{W_1}^2 + \frac{1}{T_\ell}\|p_j\|_{W_2}^2\right)
\end{equation}
with reward weighting matrices $W_1\succ 0$ and $W_2\succeq 0$. The objective of the second term in this reward function is to promote a unique solution by regularizing the learned coefficients; typically $W_2$ is small. From this batch of measurements, we get a data sample $\big(p_0, \mathcal J(p_0)\big)$. With the $j$-th iteration dataset $$\mathcal D_j=\big(p_{0:j-1}, \mathcal J(p_{0:j-1})\big),$$ we use Bayesian optimization to compute a solution to the problem
\begin{equation}
p_\infty = \arg\max_{p\in\mathbb P} \mathcal J(p),
\end{equation}
where $\mathbb P$ is an admissible compact and convex set of parameters. 

Since the gradient of $\costfn$ cannot be evaluated analytically, and the data could have inherent noise, one cannot use standard gradient-based tools to compute $p_\infty$. Instead, we resort to learning a surrogate model of $\costfn$ using Gaussian process (GP) regression~\cite{williams2006gaussian}, and exploiting the statistics of the learned surrogate to inform exploration and exploitation; that is, how to choose $p_{j}$ based on the current dataset $\mathcal D_{j-1}$. In particular, we use the expected-improvement (EI) method for acquiring subsequent $p_j$ values. The combination of GP modeling and EI acquisition is referred to herein as GP-EI, and is the topic of the next subsection.

\subsection{Bayesian Optimization with the GP-EI algorithm}
We formally reiterate our assumptions on $\mathcal J$ and $p_\star$ in the following assumption.
\begin{assumption}
\label{asmp:cost_bo}
The reward function $\costfn$ is continuous with respect to its argument $p$ for every $p\in\mathbb P$. Furthermore, the function has a unique global maximizer $p_\star$ on the set $\mathbb P$.
\end{assumption}
Since $\costfn$ is assumed continuous, we leverage the data at the $j$-th iteration to construct a surrogate GP model of the reward, given by 
\begin{equation}\label{eq:acq_fn}
\hat\costfn_j := \mathsf{GP}\left(\mu(p; \mathcal D_{j}), \sigma(p, p'; \mathcal D_{j})\right),
\end{equation}
where $\mu(\cdot)$ is the predictive mean function, and $\sigma(\cdot,\cdot)$ is the predictive variance function. Typically, the variance is expressed through the use of kernels. A commonly used kernel is the squared exponential (SE) kernel
\begin{equation}
\label{eq:kernel_SE}
\mathcal K_{\sf SE} = \sigma_0^2\exp\left(-\frac{1}{2}r^2\right),
\end{equation}
with 
\[
r^2 \equiv r^2(p,p') = \frac{\|p - p'\|^2}{\sigma_1^2}
\]
and hyperparameters $\sigma_0$ (the output variance) and $\sigma_1$ (the length scale). However, the SE kernel sometimes results in overtly smooth functions; to avoid this, another class of kernels called Mat\'ern kernels have gained popularity, of which the Mat\'ern 5/2 kernel has the form
\begin{equation}
\label{eq:kernel_matern}
\mathcal K_{\sf M52} = \sigma_0\left(1 + \sqrt{5}r + \frac{5}{3}r^2\right)\exp\left(-\sqrt{5}r\right).
\end{equation}
At the $j$-th learning iteration, for a new query sample $p\in\mathbb P$, the GP model predicts the mean and variance of the reward to be
\begin{align*}
\mu(p) &= k_j(p)^\top K_{j-1}^{-1} \mathcal J_{0:j-1}\\
\sigma(p) &= \mathcal K(p,p) - k_j(p) K^{-1}_{j-1} k_j(p)^\top,
\end{align*}
where
\begin{align*}
k_j(p) &= \begin{bmatrix}
\mathcal K(p_0, p) & \mathcal K(p_1, p) & \cdots & \mathcal K(p_{j-1}, p)
\end{bmatrix},\\
K_{j-1} &= \begin{bmatrix}
\mathcal K(p_0, p_{0}) & \cdots & \mathcal K(p_0, p_{j-1}) \\
\vdots & \ddots & \vdots \\
\mathcal K(p_{j-1}, p_{0}) & \cdots & \mathcal K(p_{j-1}, p_{j-1})
\end{bmatrix}.
\end{align*}

In Bayesian optimization, we use the mean and variance of the surrogate model $\hat\costfn_j$ in~\eqref{eq:acq_fn} to construct an acquisition function to inform the selection of a $p_{j}$ that maximizes the likelihood of improving the current best reward 
\begin{equation*}\label{eq:bestrewardk}
\bestrewardk := \max_{p\in\mathbb P} \hat\costfn_j(p).
\end{equation*}
To this end, we define an improvement function 
\[
\mathcal I := \mathcal I(p, j) = \max\{0, \costfn(p) - \bestrewardk\},
\]
whose likelihood function, based on a Gaussian posterior distribution $\mathcal N(\mu,\sigma^2)$, is given by
\[
\mathsf{L}(\mathcal I) = \frac{1}{\sqrt{2\pi} \sigma(p)}\exp\left(-\frac{1}{2}\frac{\left(\mu(p)-\bestrewardk-\mathcal I\right)^2}{\sigma^2(p)}\right).
\]
Taking an expectation of the likelihood function yields
\begin{align*}
\mathsf{EI}(p, j) &= \mathsf{E}(\mathcal I(p,j)) \\
&= \int_0^\infty \frac{\mathcal I}{\sqrt{2\pi} \sigma(p)}\exp\left(-\frac{1}{2}\frac{\left(\mu(p)-\bestrewardk-\mathcal I\right)^2}{\sigma^2(p)}\right)\, d\mathcal I.
\end{align*}
Performing a change of variables and integrating by parts yields
\begin{equation}
\mathsf{EI}(p, j) = \begin{cases} \sigma(p)\gamma(z) + (\mu(p) - \bestrewardk)\Gamma(z), & \text{if $\sigma(p) > 0$}, \\
0 & \text{if $\sigma(p) = 0$}.
\end{cases}\notag
\end{equation}
where
\[
z = \frac{\mu(p) - \bestrewardk}{\sigma(p)},
\]
and $\gamma(\cdot)$, $\Gamma(\cdot)$ are the PDF and the CDF of the mean-zero unit-variance normal distribution, respectively. 

In the $j$-th iteration of learning, we use the data $\mathcal D_j$ to construct the EI acquisition function using the surrogate $\hat\costfn_j$. Subsequently, we sample on $\mathbb P$ and obtain 
\begin{equation}
\label{eq:update_p}
p_j  = \arg\max_{p\in\mathbb P}\; \mathsf{EI}(p, j),
\end{equation}
which serves as the estimate of the model coefficients $p_t$ in~\eqref{eq:obs} until the next learning iteration.
We terminate the learning algorithm when
\begin{equation}
    \label{eq:term_crit}
    \mathsf{EI}(p,j) < \varepsilon_{\sf EI}, \;\forall\, p\in\mathbb P
\end{equation} 
for some small threshold $\varepsilon_{\sf EI}>0$. The terminal set of coefficients is defined $p_\infty$.
\begin{remark}
We do not explicitly write a hyperparameter selection procedure as it is beyond the scope of this work. Standard methods such as log-marginal-likelihood maximization is used for implementation purposes to find good GP variances and length scales. For more details, we refer the interested reader to~\cite[Chapter 5]{williams2006gaussian}.
\end{remark}
\subsection{Regret analysis}
In this subsection, we quantify the performance of the GP-EI learning algorithm in terms of the cumulative regret
\begin{equation}
\label{eq:cregret}
\mathcal R_{N} := \sum_{j=0}^{N} \costfn(p_\star) - \costfn(p_j),
\end{equation}
where $N\in\mathbb N$ denotes the number of learning/training iterations. Specifically, we will demonstrate that the regret associated with GP-EI in our learning-based observer is sublinear, and therefore, $\mathcal R_{N}/{N} \to 0$ as $N\to\infty$.

\begin{proposition}[\cite{Nguyen2017}]
\label{lem:sublinear_regret}
Suppose Assumption~\ref{asmp:cost_bo} holds. Let $\delta\in(0,1)$ and $\varepsilon_{\sf EI}$ be a fixed termination threshold. Assume that $\costfn$ lies in the reproducible kernel Hilbert space $\mathbb H_{\mathcal K}(\mathbb P)$ corresponding to the kernel $\mathcal K(p, p')$, and that the noise corrupting the reward function has zero mean conditioned on the noise history, and is bounded almost surely. Let $\|\costfn\|^2_{\mathcal K}\le B_{\costfn}$ and $\zeta_t := 2B_{\costfn} + 300\chi_t\log^3(t/\delta)$, where $\chi_t$ is a kernel-dependent constant depending on $t$. Then the GP-EI algorithm with termination criterion~\eqref{eq:term_crit} has a probabilistic bound
\begin{equation}\label{regret_bound}
\mathsf{Pr}\left(\mathcal R_N\le \sqrt{N\zeta_N\chi_N}\right) \ge 1-\delta
\end{equation}
on the cumulative regret over $N$ learning iterations. 
\end{proposition}

Using the results of Theorem~\ref{thm1} and Proposition~\ref{lem:sublinear_regret}, under Assumptions 1--\ref{asmp:cost_bo}, it is easy to write that the estimation error bound satisfies 
\begin{equation}\label{eq:isswgp}
\|e_{t}(\xi_0,\nu)\| \le \beta(t,\|\xi_0\|) + \gamma(\|\nu(t)\|_\infty),
\end{equation}
where $\nu(t)=\left[{e_t^p}^\top, {\epsilon_\phi}^\top\right]^\top$, and $e^{p}_{t}$ satisfies the probabilistic bound of Proposition~\ref{lem:sublinear_regret}.

Except for the case where the learning cost is assumed to be strongly locally convex in a neighborhood of $p^\star$, the regret bound (\ref{regret_bound}) in general, does not allow us to write an explicit bound on the parameter estimation error $e^{p}$. However, due to the continuity of the $\mathcal R_N $ in the variable $e^p$, one can use classical arguments from transformation theory of random variables, e.g. (\cite{Rohatgi76}, p. 68), to conclude about the boundedness of the parameter estimation error $e^p$. This, together with (\ref{eq:isswgp}), allows us to guarantee the safety, in the form of boundedness of the state estimation error $e$, of the observer during the learning iterations.

\begin{remark}
For squared exponential kernels, the constant $\chi_N$ is of the order $(\log(T))^{n_p+1}$, where $n_p$ denotes the dimensionality of $\mathbb P$. For  5/2 Mat\'ern kernels, $\chi_N$ is of the order $T^{\frac{n_p(n_p+1)}{5+n_p(n_p+1)}}\log(T)$. For both, the cumulative regret is sublinear as per the bound in Proposition~\ref{lem:sublinear_regret}.
\end{remark}

\section{Observer gain redesign}\label{sec:redesign}
Upon termination of the Bayesian optimization stage, we have learned the unmodeled components of the dynamical system to some accuracy. Therefore, one can (optionally) leverage this newly acquired model knowledge to update the observer gains from the initial design. 

Let $p_\infty$ denote the value of $p_t$ which minimizes simple regret over all learning iterations.
We propose a redesigned observer of the form
\begin{equation}
\label{eq:obs_redesign}
\hat x_{t+1}= A\hat x_t +Bu_t + p_\infty\psi(\hat q_t) + L(C\hat x_t - y_t),
\end{equation}
where the redesigned observer gain is $L$. The error dynamics of this redesigned observer can be rewritten as
\begin{equation}\label{eq:obs_error_redesign}
e_{t+1} = (A+LC)e_t + p_\infty^\top\Delta \psi + \psi(q_t)^\top e_\infty^p - \epsilon_\phi(q_t),
\end{equation}
where $e_\infty^p = p_\infty - p_\star$ is a constant parameter estimation error.
In this case, the estimated value $p_\infty$ will be used to redesign the observer gains.
The following theorem encapsulates the redesign conditions.
\begin{theorem}\label{thm:redesign}
If there exist matrices $K$, $P=P^\top\succ 0$, and $Q=Q^\top\succ 0$ such that

\begin{align}
\begin{bmatrix}\label{eq:thm_redesign}
-P + Q + \lpsi^2 C_q^\top C_q & \ast & \ast \\
PA + KC & -P & 0  \\
(PA + KC) p_\infty^\top & 0 & p_\infty P p_\infty^\top - I
\end{bmatrix}&\preceq 0,
\end{align}
then the redesigned observer~\eqref{eq:obs_redesign} with gain $L=P^{-1}K$ makes the error dynamics~\eqref{eq:ester1} L-ISS with respect to $[{e^p_\infty}^{\top},\;{\epsilon_{\phi}}^{\top}]^\top$.
\end{theorem}
\begin{proof}
Let $V_t=e_t^\top P e_t$ and replace $PA +KC$ by $P(A+ LC)$ by substituting $L=P^{-1}K$. Then taking successive Schur complements of the inequality in~\eqref{eq:thm_redesign} yields
\begin{equation}
\label{eq:thm_redesign_pf_a}
\begin{bmatrix}
\Omega + Q & \ast \\ P(A+LC)p_\infty^\top & p_\infty P p_\infty^\top - I
\end{bmatrix} + \begin{bmatrix}
\lpsi^2 C_q^\top C_q & 0 \\ 0 & -I
\end{bmatrix}
\preceq 0,
\end{equation}
with
\[
\Omega = (A+LC)^\top P (A+LC) - P.
\]
Taking a congruence transform with $\begin{bmatrix}
e_t^\top & \Delta \psi^\top
\end{bmatrix}^\top$ yields
\begin{align*}
& e_t^\top\Omega e_t + 2e_t^\top (A+LC)^\top Pp_\infty^\top \Delta \psi + \Delta\psi^\top p_\infty P p_\infty^\top \Delta \psi \\
&\hspace{8em} + \lpsi^2 \Delta q^\top \Delta q - \Delta \psi^\top \Delta \psi \le - e_t^\top Q e_t.
\end{align*}
Since $\lpsi$ is a Lipschitz constant of $\psi$, the above inequality implies that
\begin{equation}\label{eq:pfthm3_a}
\tilde V \le - e_t^\top Q e_t,
\end{equation}
where $$\tilde V := e_t^\top\Omega e_t + 2e_t^\top (A+LC)^\top Pp_\infty^\top \Delta \psi + \Delta\psi^\top p_\infty P p_\infty^\top \Delta \psi.$$

Let $V_t = e_t^\top P e_t$. Then, 
\[
\Delta V_t = \tilde V + \sum_{\ell=1}^3 \tilde v_\ell,
\]
where
\begin{align*}
\tilde v_1 &= \left(\psi(q_t)^\top e_\infty^p - \epsilon_\phi(q_t)\right)^\top P \left(\psi(q_t)^\top e_\infty^p - \epsilon_\phi(q_t)\right), \\
\tilde v_2 &= 2 e_t^\top (A+LC)^\top P \left(\psi(q_t)^\top e_\infty^p - \epsilon_\phi(q_t)\right),\\
\tilde v_3 &= 2\Delta\psi^\top p_\infty P \left(\psi(q_t)^\top e_\infty^p - \epsilon_\phi(q_t)\right).
\end{align*}

Following similar arguments as the proof of Theorem~1, we can bound these terms as
\begin{align*}
\tilde v_1 &\le \lambda_{\max}(P)\|\psi(q_t)^\top e_\infty^p - \epsilon_\phi(q_t)\|^2 \\
&\le \lambda_{\max}(P)\left(\bar\psi^2\|e_\infty^p\|^2 + \|\epsilon_\phi\|^2\right),\\
\tilde v_2 &\le 2 \|P(A+LC)\|\|e_t\|\left(\bar\psi\|e_\infty^p\| + \|\epsilon_\phi\|\right),\\
\tilde v_3 &\le 2 p_\infty\|P\|\|C_q\|\lpsi\|e_t\|\left(\bar\psi\|e_\infty^p\| + \|\epsilon_\phi\|\right).
\end{align*}
These bounds, along with~\eqref{eq:pfthm3_a}, yields
\begin{equation*}
\label{perform_quant_1}
\Delta V_t \le - \delta_0 \| e_t\|^2+ \delta_1\|e_\infty^p\|^2 + \delta_2 \|\epsilon_\phi\|^2 + \delta_3\|e_t\|
\end{equation*}
 where $\delta_0 = \tfrac{1}{2}\lam_{\min}(Q)$, $\delta_1 = \lambda_{\max}(P)\bar{\psi}^2$, and $\delta_2=\lam_{\max}(P)$, and $\delta_3 \le 0$ if
\[
\|e_\infty^p\| \le \frac{\lambda_{\min}(Q)}{4\bar\psi\left(\|P(A+LC)\|+p_\infty\lpsi\|P\|\|C_q\| \right)}\|e_t\| - \frac{1}{\bar\psi}\|\epsilon_\phi\|.
\]
The rest of the proof is identical to the proof of Theorem~\ref{thm1}.
\end{proof}

\begin{remark}
Taking advantage of the learned coefficients $p_\infty$, we obtain simpler LMI conditions in~\eqref{eq:thm_redesign}. Furthermore, from the L-ISS definition, we deduce that re-designing the observer with a smaller parameter estimation error $e^p_\infty$ leads to a smaller state estimation error.
\end{remark}

A complete pseudocode is provided in the Appendix.

\section{Numerical Example}\label{sec:ex}
Consider the van der Pol oscillator system
\[
\dot x = \begin{bmatrix}
0 & 1 \\ -1 & 1
\end{bmatrix} x + \begin{bmatrix}
0 \\ -1
\end{bmatrix} x_1^2 x_2, \; y = x_1.
\]
We begin by Euler discretization of the continuous-time dynamics with a sampling time of $\tau=0.01$~s. Comparing with~\eqref{eq:system_description}, we have
\[
A = \begin{bmatrix} 1 & \tau \\ \tau & 1 - \tau\end{bmatrix}, \; B = \begin{bmatrix}
0 \\ -\tau
\end{bmatrix},\; C = \begin{bmatrix}
1 & 0
\end{bmatrix}, \; C_q = I,
\]
and $\phi(q) = q_1^2 q_2$. We know $A$, $B$, $C_q$, and $C$, as per Assumption~\ref{asmp:knowledge_dynamics}. Owing to limit cycle behaviour of the oscillator, Assumption~\ref{asmp:x_bdd} holds. Furthermore, the nonlinearity $\phi$ is locally Lipschitz, with unknown Lipschitz constant. Data is collected by using forward simulations of the oscillator system $T_\ell = 40$~s from an initial condition $[1, 1]^\top$.

We design the initial observer gain by solving the problem~\eqref{eq:linesearch}, assuming $\bar{p}_\star=10^{-2}$. We run a golden-section search for $\lc\in [0, 10]$, which yields an initial observer gain of 
\[
L_0 = \begin{bmatrix}
1.1727 & 7.3679
\end{bmatrix}^\top
\]
for $\lchat=4.332$. The observer initial conditions are set to zero.

\begin{figure}[!ht]
    \centering
    \includegraphics[width=\columnwidth]{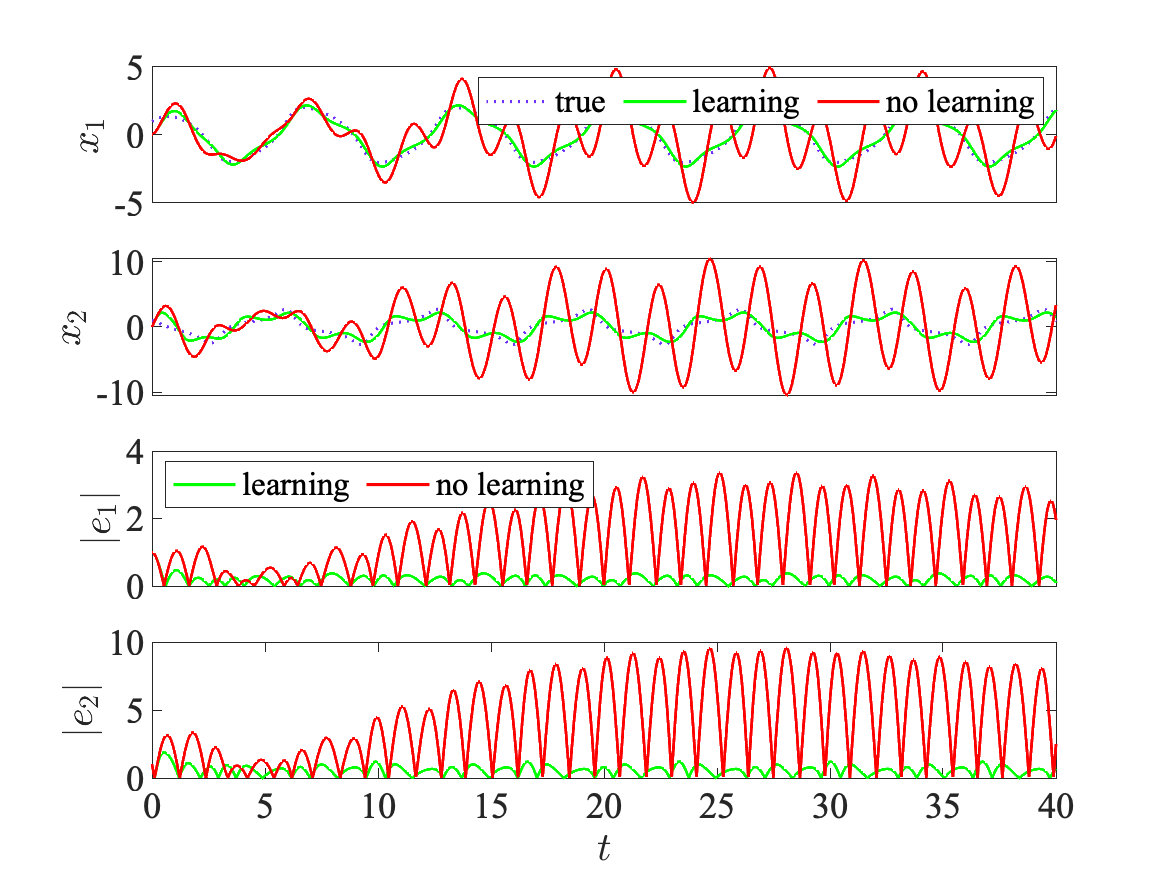}
    \caption{Comparison of state estimates and error norms with and without learning from data.}
    \label{fig:ex1_a}
\end{figure}
For fairness, we choose Legendre polynomials as basis functions and not monomials, which would perfectly fit the nonlinearity $\phi(q)$ with very few coefficients. Our basis functions are
\[
\psi(\hat q) = 10\begin{bmatrix}
(3\hat q_1^2-1)(3\hat q_2^2-1) \\
    (3 \hat q_1^2-1)\hat q_2 \\ \hat q_1(3\hat q_2^2-1) \\
    (5\hat q_1^3-3\hat q_1) \\ (5\hat q_2^3-3\hat q_2)
\end{bmatrix},
\]
where the scaling helps find good coefficients despite the strong bound on $\bar p_\star$. Performing Bayesian optimization with the GP-EI method using a weighting matrix $W_1=200$, $W_2=1$, and termination threshold $\varepsilon_{\sf EI}=0.01$ yields
\[
p_\infty = 10^{-3}\times \begin{bmatrix}
   -0.6077 \\
    8.4930 \\
   -9.2877 \\
    1.8897 \\
    9.8417
\end{bmatrix}
\]
in $N=200$ training iterations, with $1000$ uniformly random samples on $[-\bar p_\star, \bar p_\star]^5$ to select $p_j$ at each $j$ using~\eqref{eq:update_p}. The surrogate GP model is constructed at each time using sklearn in Python or MATLAB's \texttt{fitrgp} function in the Statistics and Machine Learning Toolbox. We use the Mat\'ern ARD-52 kernel and select hyperparameters using BFGS on the log-likelihood function at each training iteration. The effects of model learning is illustrated in Figure~\ref{fig:ex1_a}. Prior to learning, the initial observer design provides bounded error dynamics (red continuous line), but the error is considerably reduced after learning is completed (green continuous line). 
	
\begin{figure}[!ht]
    \centering    \includegraphics[width=\columnwidth]{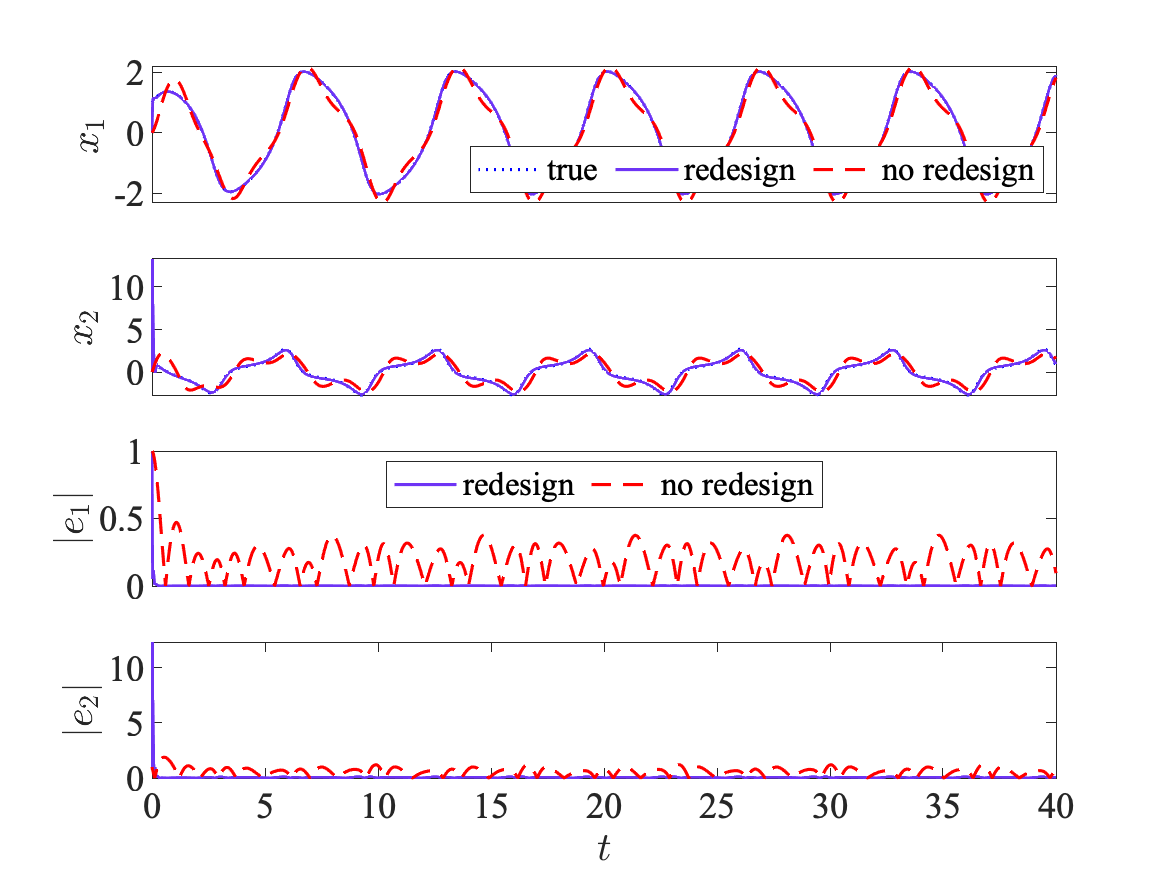}
    \caption{Comparison of state estimates and error norms with and without final redesign of the observer gains.}
    \label{fig:ex1_b}
\end{figure}	

With the learned nonlinearity in a neural approximator/basis expansion form $p_\infty^\top \psi(\hat q)$, we select $\mathbb Q = [-5, 5]^2$ and use the kernelized Lipschitz estimation method in~\cite{9042816} to compute $\lchat=24.537$, with which we can solve~\eqref{eq:thm_redesign} to get
$
L = \begin{bmatrix}
   39.3221 &
  958.7488
\end{bmatrix}^\top
$.
As seen from Figure~\ref{fig:ex1_b}, the redesign results in further reduction of state estimation error.

In order to demonstrate the effectiveness of the proposed algorithm, we compare the performance of our learning-based observer with well-known robust and adaptive observers in the literature. The observers we compare against include: the robust adaptive observer (RA) described in~\cite[Chapter 5]{ioannou2012robust}, the high-gain (HG) observer described in~\cite{khalil2014high}, and the recurrent neural network (RNN) learning-based observer proposed in~\cite{chen2017neural}. In all cases, we make multiple simplifications regarding knowledge of the model dynamics, since these observers failed to converge with completely unknown dynamics. In the first two observers, one requires a `nominal model' for $\phi(q)$, which we assume to be $\hat x_1^2\hat x_2+0.2*\sin(\hat x_2)$ to emulate a bounded disturbance on the model knowledge. For the neural observer, we start with an initial condition $\begin{bmatrix} 0.5 & 0.5 \end{bmatrix}^\top$ which is close to the true system state, because starting from the origin results in the neural approximator diverging. A comparison of the state estimation error in the unmeasured state, that is, $|e_2|$ is illustrated in Fig.~\ref{fig:comp_study}. Clearly, our proposed observer (purple line) outperforms the others despite the simplifications made to help the other observer architectures. The RA observer (amber line) demonstrates good performance, but slow convergence rate, despite extensive hand-tuning - another drawback of this observer is that the observer state-space 8-dimensional and there are multiple hyperparameters and gains that need hand-tuning. The HG observer (green line), expectedly, has high state estimation error initially, although the steady-state estimation norm is smaller than RA and RNN,  albeit larger than our proposed observer. Finally, the RNN observer (magenta line) exhibits significantly worse performance with time, and sometimes the estimation error grows large; this is because the observer estimation error dynamics are uniformly ultimately bounded, but not ISS, so large approximation errors in the recurrent neural network result in poor quality of the state estimates.
\begin{figure}[!ht]
    \centering
    \includegraphics[width=.9\columnwidth]{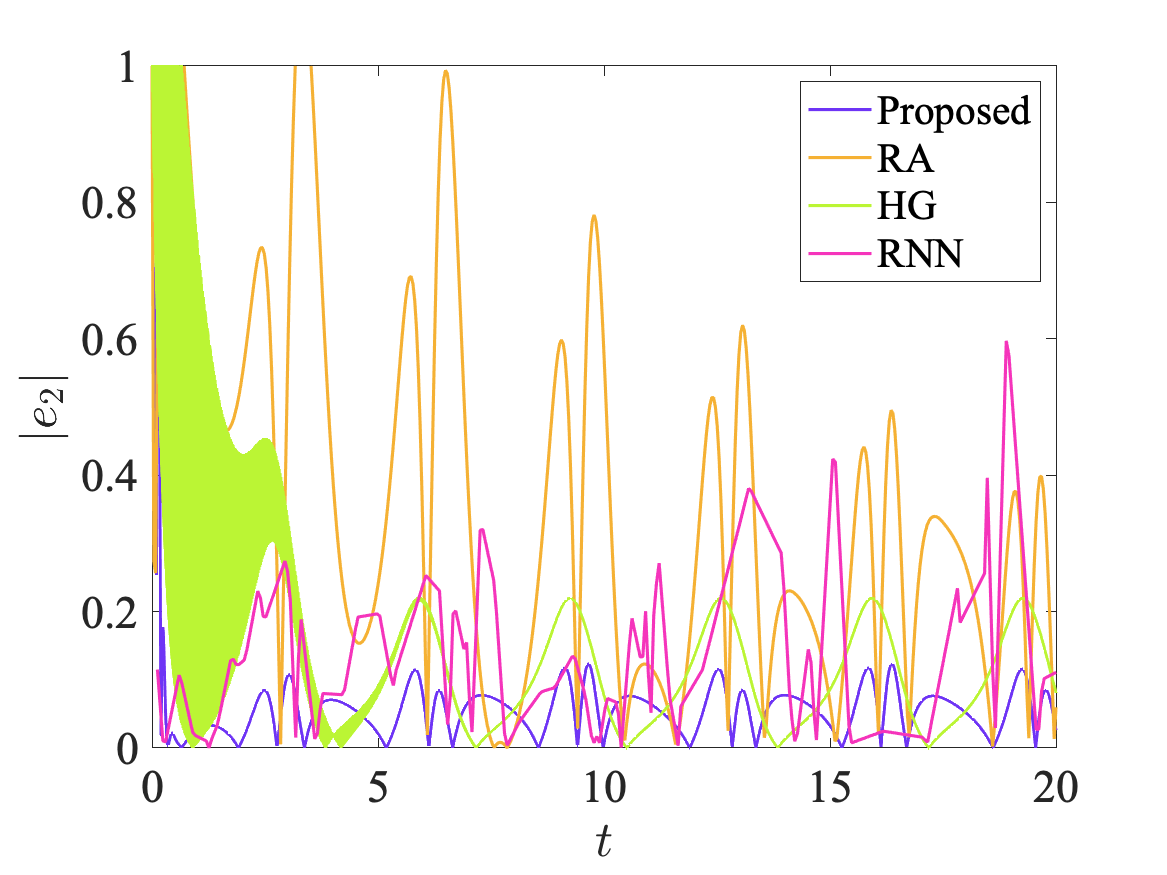}
    \caption{Comparative study of our proposed method with robust and adaptive observers in the literature. RA = Robust adaptive observer~\cite[Chapter 5]{ioannou2012robust}, HG = High-gain observer~\cite{khalil2014high}, RNN = Recurrent neural network observer~\cite{chen2017neural}. RA and HG have knowledge of the nonlinearity $\phi$, and RNN state is initialized near the system's initial state.}
    \label{fig:comp_study}
\end{figure}
\section{Conclusions}\label{sec:conc}
In this paper, we provide a design methodology for constructing state estimators with unmodeled dynamics. We generate an initial observer that is robust to the learning error, in thes sense of ISS guarantees, and use this conservative initial state estimator to iteratively learn the unmodeled dynamics via Bayesian optimization in a purely data-driven manner. Once a suitable estimate of the complete model is achieved, based on an information measure, we redesign the observer using the learned dynamics to reduce conservatism. Furthermore, the proposed modular design, based on a model-based part ensuring the safety of the learning via ISS guarantees, added to a data-driven optimization part, allows us to plug-in any data-driven optimization approach with convergence guarantees, e.g. Bayesian optimization, extremum seekers, etc.  

Further investigations will be conducted to deal with the case of output measurement noise, which can enter the system both at the observer feedback level or at the learning cost function measurement level.

\bibliographystyle{unsrt}
\bibliography{refs.bib}	

\newpage
\appendix
\section*{Appendix: Implementation}
\begin{algorithm}[!ht]
\caption{Bayesian Learning-based Observer}\label{algo:1}
\small
\begin{algorithmic}[1]
\Require Initial conditions $x_0$, $\hat x_0$
\Require System matrices $A$, $B$, $C$
\Require Function argument matrix $C_q$ \Comment {default: $I$}
\Require Basis functions $\psi$
\Require Initial coefficients $p_0$ \Comment{default: 0} 
\Require Upper bound on coefficients $\bar p_\star$ \Comment{default: $10^{-3}$}
\Require Range of $\lchat$ to perform line search
\Require Batch size for learning $T_{\ell}$
\Require Cost function weights $W_1, W_2$ \Comment{p.s.d. matrices}
\Require Kernel $\mathcal K$ for GP regression \Comment{default: Mat\'ern 52} 
\Require Termination threshold $\varepsilon_{\sf EI}$
\Comment{default: $10^{-3}$} 
\Algphase{Initial L-ISS observer design}
\State Select $A$ such that $(A,C)$ is observable
\State $\lchat, L_0\leftarrow$ perform line search~\eqref{eq:linesearch} \Comment{involves solving SDP~\eqref{eq:thm2_all} for $L_0$ corresponding to $\lchat$}
\State Parameterize initial observer~\eqref{eq:obs} with $L_0$ and $p_0$
\Algphase{Safe learning from online data}
\For{$j=0:N-1$}
\State Collect batch data (estimated states $\{\hat x\}$ and measurements $\{y\}$) from system
\State $\mathcal J(p_j)\leftarrow$ compute via~\eqref{eq:costfn_comp}
\State Learn surrogate GP model from data $\{p_j, \mathcal J(p_j)\}$
\State $\mu,\sigma\leftarrow$ use GP to compute on samples drawn from $\mathbb P$
\State $p_{j+1}\leftarrow$ obtain via expected improvement~\eqref{eq:update_p}
\State Re-parameterize observer~\eqref{eq:obs} with $L_0$ and $p_t\leftarrow p_{j+1}$
\State Terminate if condition~\eqref{eq:term_crit} holds
\EndFor
\State $p_\infty\leftarrow$ terminal value of $p_{j+1}$
\Algphase{Observer gain redesign}
\State $L\leftarrow$ redesigned observer gain in~\eqref{eq:thm_redesign}
\State Parameterize re-designed observer~\eqref{eq:obs_redesign} with $L$ and $p_\infty$
\end{algorithmic}
\end{algorithm}
\end{document}